\documentclass[lettersize,journal]{IEEEtran}
\usepackage{amssymb}
\usepackage{amsmath}
\usepackage{cite}
\usepackage{url}
\usepackage{xcolor}
\usepackage{cite,graphicx,amsmath,amssymb}
\usepackage{fancyhdr}
\usepackage{mdwmath}
\usepackage{mdwtab}
\usepackage{caption}
\usepackage{amsthm}
\usepackage{setspace}
\usepackage{hyperref}
\usepackage{algorithm}
\usepackage{algorithmic}
\usepackage{multirow}
\usepackage{makecell}
\usepackage{cuted}
\usepackage{mathtools}
\usepackage{subcaption}
\usepackage{mathrsfs}
\usepackage{bm}
\usepackage{pifont}

\hypersetup{colorlinks=false}
\usepackage{stfloats}
\usepackage{titlesec}
\expandafter\def\expandafter\normalsize\expandafter{%
    \normalsize%
    \setlength\abovedisplayskip{3pt}%
    \setlength\belowdisplayskip{3pt}%
    \setlength\abovedisplayshortskip{1pt}%
    \setlength\belowdisplayshortskip{1pt}%
}

\allowdisplaybreaks[3]

\newtheorem{remark}{Remark}
\newtheorem{theorem}{Theorem}

\newtheorem{lemma}{Lemma}

\newtheorem{corollary}{Corollary}

\newtheorem{proposition}{Proposition}

\allowdisplaybreaks

\makeatletter
\renewenvironment{proof}[1][\proofname]{%
  \par\pushQED{\qed}%
  \normalfont \topsep=2pt \partopsep=0pt
  \trivlist
  \item[\hskip\labelsep
        \itshape #1\@addpunct{.}]\ignorespaces
}{%
  \popQED\endtrivlist\@endpefalse
}
\makeatother
\begin{document}
\title{Center-Fed Pinching Antenna System for Uplink Environment Sensing}
\vspace{-0.5cm}
\author{%
Cong Yu,
Xu Gan,~\IEEEmembership{Member,~IEEE},
Bo Ai,~\IEEEmembership{Fellow,~IEEE},\\
Yuanwei Liu,~\IEEEmembership{Fellow,~IEEE},
Wei Chen,~\IEEEmembership{Senior Member,~IEEE}

\thanks{ 
This work was supported in part by the Mobile Information Networks-National Science and Technology Major Project(2026ZD1305700);  the Natural Science Foundation of China (U2468201, 62221001); ZTE Industry-University-Institute Cooperation Funds under Grant No. IA20250115003-PO0006.  (\textit{Corresponding authors: Bo Ai, Wei Chen
.}) Cong Yu, Bo Ai and Wei Chen are with School of Electronic
 and Information Engineering, Beijing Jiaotong University, Beijing 100044,
 China. (e-mail: \{cong\_yu, boai, weich\}@bjtu.edu.cn). Xu Gan and Yuanwei Liu are with the Department of Electrical and Electronic Engineering, The University of Hong Kong, Hong Kong. (e-mail: \{eee.ganxu, yuanwei\}@hku.hk).
\vspace{-0.5em}
 }}

\maketitle
\pagestyle{empty}
\thispagestyle{empty}
\renewcommand{\algorithmicrequire}{\textbf{Input:}}
\renewcommand{\algorithmicensure}{\textbf{Output:}}
\vspace{-0.5cm}
\begin{abstract}
A center-fed pinching antenna system (C-PASS)-enabled uplink environment sensing framework is proposed. Through the center-fed framework, doubled degrees of freedom is achieved compared to conventional end-fed PASS. Based on this, we consider an uplink sensing scenario, in which a linear inverse model is developed to reconstruct the environment through signals scattered by the environment object. In the proposed framework, the distance between the feed points for stable separation of the received signals is characterized in closed form. Furthermore, Ziv-Zakai bound (ZZB) expressions for the mean-squared reconstruction error are derived for C-PASS and end-fed PASS. Based on these theoretical results, it can be proved that C-PASS achieves a strictly lower reconstruction error bound than conventional PASS for uplink environment sensing. Finally, numerical results validate the accuracy of the derived ZZB expressions and 1) demonstrate that C-PASS provides more stable separation of the received signals, and 2) confirm the consistent performance advantages of C-PASS.
\end{abstract}

\begin{IEEEkeywords}
Center-fed pinching antenna, Received-signal separation, Uplink environment sensing, Ziv-Zakai bound.
\end{IEEEkeywords}
\section{Introduction}
Pinching antenna systems (PASS) have recently attracted considerable attention as a reconfigurable architecture for wireless communications~\cite{pass1,pass2,sen1}. By guiding electromagnetic signals inside a dielectric waveguide and radiating them through movable pinching antennas (PAs), PASS shortens the free-space path and provides spatial diversity along the same structure~\cite{pass2}. These features are useful not only for communication, but also for localization~\cite{sen1}, scattering-profile reconstruction~\cite{sen2}, and environment sensing~\cite{env}. Among these applications, uplink environment sensing is particularly attractive because the scene can be inferred directly from user transmissions without a dedicated downlink probing stage.

However, most existing PASS sensing studies focus on downlink scenarios~\cite{pass1,pass2,sen1,sen2}, while uplink studies mainly consider communication metrics such as channel estimation~\cite{sen_up2}, capacity bounds~\cite{sen_up}{, and over-the-air federated learning~\cite{feb}}. Thus, uplink environment sensing with PASS remains underexplored. A key difficulty lies in the conventional end-fed architecture, where the uplink signals captured by multiple PAs propagate along the waveguide and are coherently combined at a single feed port. Consequently, each user-slot block provides only one independent observation, which limits the sensing degrees of freedom (DoF) to 1 and makes the received measurement matrix structurally ill-conditioned~\cite{sen_li}. To address this limitation, the recently proposed center-fed PASS (C-PASS)~\cite{gan2025c} provides a new opportunity. By incorporating tunable T-junction waveguide splitters~\cite{gan2026center}, C-PASS introduces multiple feed points at different waveguide locations, each connected to an independent RF chain. The forward- and backward-PA-received signals then travel different waveguide distances to each feed point, allowing C-PASS to avoid the structural rank deficiency. Nevertheless, prior C-PASS work has focused exclusively on communication problems, including DoF characterization~\cite{gan2025c}, beamforming design~\cite{gan2026center}, and capacity analysis~\cite{gan2026center2}. Its potential for uplink environment sensing remains open.

The main contributions are summarized as follows: i) A C-PASS-enabled uplink environment sensing framework is developed, in which the indoor scene is discretized into binary voxels and the scattered uplink signals are formulated as a linear inverse model for reconstructing the environment. ii) The distance between the feed points for stable separation of the received signals is characterized in closed form. iii) Closed-form Ziv-Zakai bound (ZZB) expressions are derived for both C-PASS and end-fed PASS, revealing that C-PASS achieves a strictly lower reconstruction error bound. Numerical results validate the derived ZZB and confirm that this advantage holds consistently across different numbers of users and time slots.

\section{System Model}
\label{sec:II}
As illustrated in Fig.~\ref{fig:sys}, a C-PASS-enabled uplink environment sensing system is considered in an indoor scenario. A dielectric waveguide is mounted at height $D_h$ along the $y$-axis. Its geometric center is $\boldsymbol{\ell}_c=[x_c,y_c,D_h]^T$, and two feed points are placed symmetrically about the waveguide center at
\begin{equation}
\boldsymbol{\ell}_{c,1} = \left[x_c,\; y_c - \tfrac{d}{2},\; D_h\right]^T, 
    \boldsymbol{\ell}_{c,2} = \left[x_c,\; y_c + \tfrac{d}{2},\; D_h\right]^T,
    \label{eq:cpass_feeds}
\end{equation}
where $d > 0$ denotes the distance between the two feed points, with $y_{c,1} = y_c - \tfrac{d}{2}$ and $y_{c,2} = y_c + \tfrac{d}{2}$. At slot $t\in\{1,\ldots,T\}$, one forward PA and one backward PA are activated at $\boldsymbol{\ell}_t^F=[x_c,\,y_t^F,\,D_h]^T$ and $\boldsymbol{\ell}_t^B=[x_c,\,y_t^B,\,D_h]^T$, where $y_t^F=y_c-d_t^F$ and $y_t^B=y_c+d_t^B$. The system is served by $U$ single-antenna UEs at known positions $\{\boldsymbol{\ell}_u\}_{u=1}^{U}$, where $\boldsymbol{\ell}_u = [x_u, y_u, 0]^T$. The $u$-th user transmits a pilot $s_u(t)$ with power $P_u$, and the orthogonality condition $\mathbb{E}[s_u(t)s_v^*(t)]=\delta_{uv}$ enables matched-filter separation at the receiver for all $u$. It can be assumed that the UE positions and pilots are known, thus the deterministic line-of-sight (LoS) component can be estimated and removed. 

The region of interest (RoI) is a cuboid of dimensions $S_x \times S_y \times D_h$, discretized into $N$ voxels of size $s_x \times s_y \times d_h$. The sensing objective is to determine whether a scatterer is present in each voxel, characterized by $x_n\in\{0,1\}$, where $x_n = 1$ denotes the presence of a scatterer at voxel center $\boldsymbol{\ell}_n = [\hat{x}_n, \hat{y}_n, \hat{z}_n]^T$ and $x_n = 0$ denotes absence. The complete environmental profile is thus represented by $\mathbf{x} = [x_1, x_2, \ldots, x_{N}]^T \in \mathbb{R}^{N}$. Then, the goal of environment sensing is to recover $\mathbf{x}$ from uplink received signals. For a scatterer at voxel $n$, the two-hop propagation factor from $u$-th UE to the forward PA at slot $t$ is
\begin{equation}
    \xi_{t,u}^F(n) = \eta^2 \frac{\exp\!\left(-j\beta_0(r_{u,n} + r_{n,t}^F)\right)}{r_{u,n}\, r_{n,t}^F},
    \label{eq:xi_F}
\end{equation}
and similarly for the backward PA with $\xi_{t,u}^B(n)\!\! = \eta^2\!\exp(-\!j\beta_0(r_{u,n}\!+\!r_{n,t}^B))/(r_{u,n} r_{n,t}^B)$, where $r_{u,n}\!\! = \!\!\|\boldsymbol{\ell}_u \!\!-\! \boldsymbol{\ell}_n\|  _2$, $r_{n,t}^F \!\!=\!\! \|\boldsymbol{\ell}_n\!\! -\!\boldsymbol{\ell}_t^F\|_2$, $r_{n,t}^B \!= \!\!\|\boldsymbol{\ell}_n \!\!-\! \boldsymbol{\ell}_t^B\|_2$, $\beta_0 = 2\pi/\lambda$ is the free-space wavenumber, and $\eta$ denotes the path-loss coefficient.
\begin{figure}[!t]
  \centering
  \includegraphics[width=0.75\linewidth]{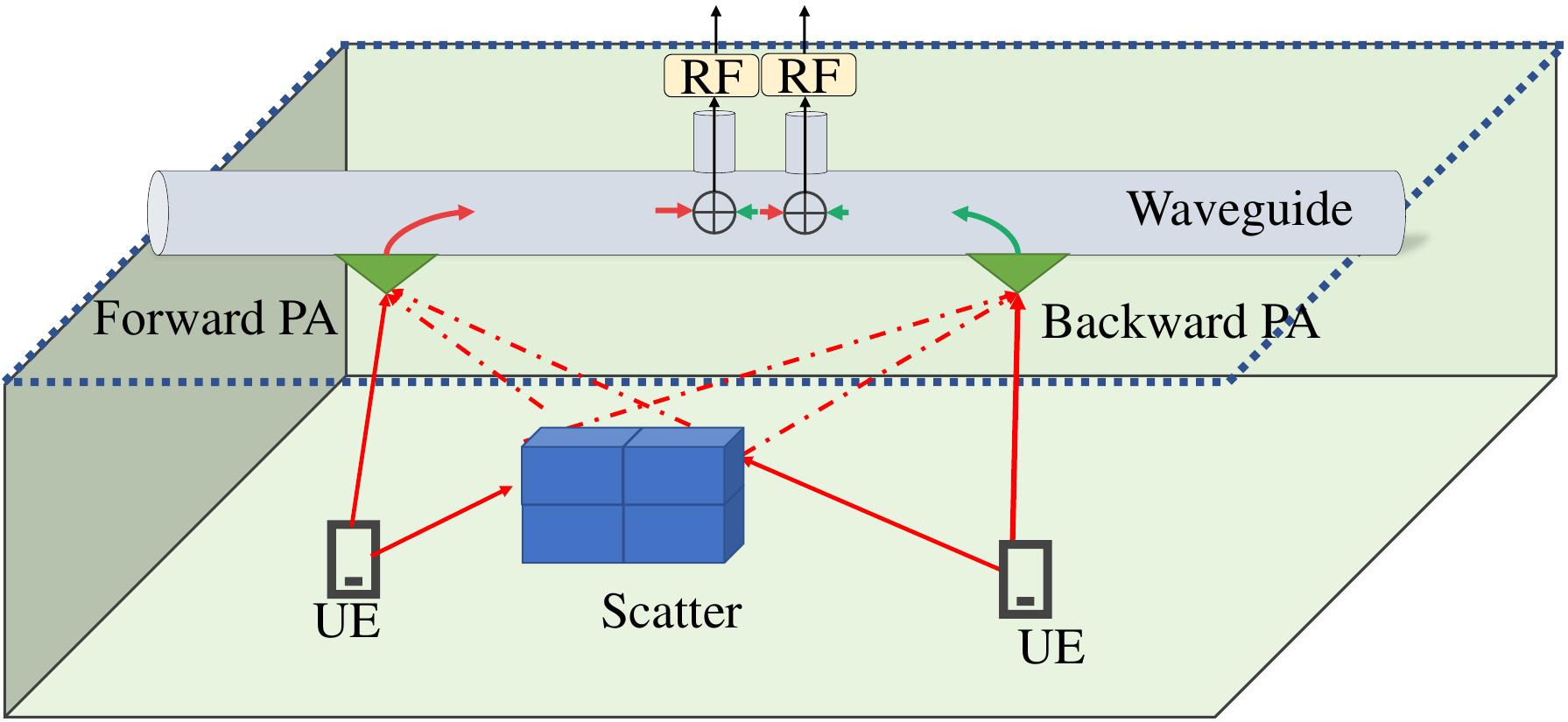}
  \caption{System model of the proposed framework. Solid and dashed red arrows denote LoS and scatter-reflected paths. The red and green arrows inside the waveguide are received signals from different PAs.}
  \label{fig:sys}
\end{figure}
The signals received at the two PAs travel along the waveguide to each feed point with propagation constant $\beta_g\!\!=\!\!2\pi n_e/\lambda$, where $n_e$ is the effective refractive index. {Thus,  $\lambda_g\triangleq\lambda/n_e=2\pi/\beta_g$ denotes the guided wavelength.} Accounting for the free-space scattering response in \eqref{eq:xi_F} and the subsequent waveguide propagation, the composite channel coefficient from $u$-th UE to feed point $m\!\!\in\!\!\{1,2\}$ via voxel $n$ at slot $t$ is
\begin{equation}
\begin{aligned}
    \!\!&A_{t,u,m}(n)\!\! \!\\ &=\!\! \sqrt{P_u}\!\!\left(\!
    e^{-j\beta_g|y_t^F - y_{c,m}|}\,\xi_{t,u}^F(n)
    \!+\!
    e^{-j\beta_g|y_t^B - y_{c,m}|}\,\!\xi_{t,u}^B(n)\!
    \right),
    \label{eq:coeff}
\end{aligned}
\end{equation}
where $y_t^F = y_c - d_t^F$ and $y_t^B = y_c + d_t^B$ are the $y$-coordinates of the forward and backward PAs, respectively.

After matched filtering against the pilot of $u$-th UE, the received signal at feed point $m$ is
\begin{equation}
    \tilde{y}_{m,u}(t) = \sum_{n=1}^{N} A_{t,u,m}(n)\, x_n + \tilde{n}_{m,u}(t),
    \label{eq:obs}
\end{equation}
where $\tilde{n}_{m,u}(t) \sim \mathcal{CN}(0,\sigma^2)$ is additive white Gaussian noise (AWGN). {The noise samples are assumed independent across feed points, users, and time slots.} Stacking the scalar received signals $\tilde{y}_{m,u}(t)$ over $m \in \{1,2\}$, $u \in \{1,\ldots,U\}$, and $t \in \{1,\ldots,T\}$ into a vector 
$\tilde{\mathbf{y}} = [\tilde{y}_{1,1}(1),\tilde{y}_{2,1}(1),\ldots, \tilde{y}_{2,U}(T)]^T \in \mathbb{C}^{2UT}$ yields the linear measurement model
\begin{equation}
    \tilde{\mathbf{y}} = \mathbf{A}\,\mathbf{x} + \tilde{\mathbf{n}}.
    \label{eq:global}
\end{equation}
Let $i(m,u,t)=2U(t-1)+2(u-1)+m$ denote the row index associated with feed point $m$, the $u$-th UE, and slot t. Then the measurement matrix satisfies $[A]_{i(m,u,t),n}=A_{t,u,m}(n)$, and the corresponding noise entry is $[\tilde n]_{i(m,u,t)}=\tilde n_{m,u}(t)$. {If the ideal assumptions are relaxed, residual LoS and UE-location errors respectively introduce a nuisance term $\mathbf b_{\rm LoS}$ and a matrix perturbation $\Delta\mathbf A$, while nonorthogonal pilots produce an effective covariance $\mathbf R$; after whitening, the pairwise distance becomes $\boldsymbol\delta^H\mathbf A^H\mathbf R^{-1}\mathbf A\boldsymbol\delta$.}

\section{Performance Analysis for Environment Sensing}

To evaluate the sensing performance advantages of the proposed C-PASS architecture, this section provides a comparative analysis of the sensing performance with conventional end-fed PASS.

\subsection{System Configuration}
For a fair comparison, we consider two architectures configured as follows:

\begin{itemize}
\item \textbf{Center-Fed PASS}: As shown in Section~\ref{sec:II}, at each slot $t\in\{1,\ldots,T\}$, one forward PA and one backward PA are activated. The offsets are set as $d_t^F\!\!=\!\!d_t^B\!\!=\frac{S_y}{2}\!\frac{t-1}{T-1}$.

\item \textbf{End-Fed PASS}: The end-fed PASS is equipped with two feed points co-located at the left end of the waveguide, i.e., $y_{e,1}=y_{e,2}=0$. The corresponding composite channel coefficients therefore satisfy
$A_{t,u,m}^{\mathrm{PASS}}(n)
\!\!=\!\!\sqrt{P_u}\Big(\!
e^{-j\beta_g|y_t^F|}\xi_{t,u}^F(n)
\!\!+\!\!e^{-j\beta_g|y_t^B|}\xi_{t,u}^B(n)
\!\Big)$ for $m\in\{1,2\}$.
\end{itemize}

\subsection{Received-Signal Separation Analysis}
This subsection analyzes how C-PASS separates the received signals from the forward and backward PAs. In end-fed PASS, the two PA-received signals are coherently combined into one scalar signal in each user-slot block, yielding only one sensing DoF. In contrast, C-PASS collects two feed-point measurements with different mixing coefficients, increasing the DoF to 2. However, a higher DoF alone does not guarantee that the two received signals can be stably separated in the presence of noise. To formalize this separation problem, for each user-slot block, we factor the received measurement matrix from \eqref{eq:global} as $\mathbf{A}_{t,u}=\mathbf{W}_t\mathbf{S}_{t,u}$, with $\mathbf{S}_{t,u}\in\mathbb{C}^{2\times N}$ collecting the two-hop propagation responses from all voxels to the forward and backward PAs, i.e.,
\begin{equation}
\mathbf{S}_{t,u}=\begin{bmatrix}
\xi_{t,u}^F(1) & \cdots & \xi_{t,u}^F(N)\\
\xi_{t,u}^B(1) & \cdots & \xi_{t,u}^B(N)
\end{bmatrix},
\label{eq:response_matrix}
\end{equation}
and $\mathbf{W}_t\in\mathbb{C}^{2\times2}$ maps the forward- and backward-PA responses to the two feed-point measurements:
\begin{equation}
\mathbf{W}_t =
\begin{bmatrix}
e^{-j\beta_g|y_t^F-y_{c,1}|} & e^{-j\beta_g|y_t^B-y_{c,1}|}\\
e^{-j\beta_g|y_t^F-y_{c,2}|} & e^{-j\beta_g|y_t^B-y_{c,2}|}
\end{bmatrix}.
\label{eq:mixing}
\end{equation}
{For the considered geometry, $\mathbf S_{t,u}$ has full row rank and introduces no structural rank loss; hence, the architecture-dependent rank distinction is governed by $\mathbf W_t$.}

We then quantify the separation stability by the condition number $\kappa=\sigma_{\max}/\sigma_{\min}$ of $\mathbf{W}_t$, since a small $\kappa$ indicates that the two PA-received signals can be stably separated from the measurements, whereas a large $\kappa$ means that the two PA-received signals are nearly indistinguishable. Based on this formulation, the following proposition characterizes the conditioning of $\mathbf{W}_t$ as a function of $d$. Substituting~\eqref{eq:cpass_feeds} into~\eqref{eq:mixing}, the determinant is
\begin{equation}
\det\!\left(\mathbf{W}_t^{\mathrm{Cen}}\right) = 2j\sin(\beta_g d)\,e^{-j\beta_g(\cdot)}.
\label{eq:det_cpass}
\end{equation}
Since every entry has unit magnitude, $\sigma_{\max}^2+\sigma_{\min}^2=\|\mathbf{W}_t^{\mathrm{Cen}}\|_F^2=4$, while~\eqref{eq:det_cpass} gives $\sigma_{\max}\sigma_{\min}=2|\sin(\beta_g d)|$. Solving for the two singular values shows that the condition number depends on $d$ as
\begin{equation}
    \kappa\!\left(\mathbf{W}_t^{\mathrm{Cen}}\right)
    = \left|\cot\!\left(\frac{\beta_g d}{2}\right)\right|,
    \label{eq:cond}
\end{equation}
which is minimized when $d=(1+2S)\lambda_g/4$, $S\in\mathbb{Z}$. At these spacings, $\kappa(\mathbf{W}_t^{\mathrm{Cen}})=1$, so equivalently, the two columns of $\mathbf{W}_t^{\mathrm{Cen}}$, corresponding to the forward and backward PA-received signals, are orthogonal. Hence, the two PA-received signals can be separated from the feed-point measurements with minimum sensitivity to noise. By contrast, $\mathbf{W}_t^{\mathrm{End}}$ is rank deficient and $\kappa(\mathbf{W}_t^{\mathrm{End}})\to\infty$ for any PA configuration.

\begin{remark}
{\emph{The factorization also clarifies the hardware tradeoff. For end-fed PASS, $\mathbf W_t^{\mathrm{End}}=\mathbf 1_2\mathbf w_t^T$ has rank one, so a second co-located RF chain adds no sensing DoF; at an optimal C-PASS spacing, $(\mathbf W_t^{\mathrm{Cen}})^H\mathbf W_t^{\mathrm{Cen}}=2\mathbf I$, so two spatially separated feeds connected to calibrated RF chains provide independent simultaneous mixtures, whereas a one-chain end-fed receiver requires serial pilot subslots. For $M$ feeds and $P$ PA components, $\mathbf A_{t,u}=\mathbf W_t\mathbf S_{t,u}$ with $\mathbf W_t\in\mathbb C^{M\times P}$ and $\mathbf S_{t,u}\in\mathbb C^{P\times N}$; if $M\ge P$ and $\mathbf W_t$ has full column rank, $\kappa(\mathbf W_t)$ remains the separability criterion, while stacked multi-waveguide blocks preserve $\|\mathbf A\boldsymbol\delta\|_2^2$. The two-feed spacing rule follows from the $2\times2$ determinant in~\eqref{eq:det_cpass} and becomes a joint feed-placement optimization of $\sigma_{\min}(\mathbf W_t)$ or $\kappa(\mathbf W_t)$. This hardware overhead is justified by the improved measurement diversity and reduced ZZB under the same sensing budget.}}
\end{remark}

\subsection{Fundamental Performance Limit via ZZB}
\label{31}

In this subsection, we characterize the fundamental sensing accuracy of the proposed framework. Since the sensing task is a discrete detection problem, the Cram\'er-Rao bound (CRB) is ill-suited here, as it applies to unbiased estimators of continuous parameters and provides no meaningful characterization near the recovery threshold. The ZZB, by contrast, remains valid across the full SNR range and directly accounts for the discrete nature of support recovery~\cite{jeong2025zzb}. In the following, a closed-form ZZB expression is derived as a function of the measurement matrix $\mathbf{A}$ and noise variance $\sigma^2$.

To define the ZZB, we set $K$ as the number of active voxels, restrict the binary scene vector to the $K$-sparse support set $\mathcal{X}_K \triangleq \left\{\mathbf{x}\in\{0,1\}^N:\mathbf{1}^T\mathbf{x}=K\right\}$. Let $\hat{\mathbf{x}}\in\mathcal{X}_K$ denote any estimate of the scene vector. Define
$J(\hat{\mathbf{x}},\mathbf{x}) \triangleq \frac{1}{2}\|\hat{\mathbf{x}}-\mathbf{x}\|_0 \in \{0,1,\ldots,K\}$ as the number of missed active voxels. Since both vectors have the same $K$,
\begin{equation}
\frac{1}{K}\|\hat{\mathbf{x}}-\mathbf{x}\|_2^2=\frac{2}{K}J(\hat{\mathbf{x}},\mathbf{x}).
\label{eq:nmse_relation}
\end{equation}
So a lower bound on $J/K$ is also a lower bound on the mean squared error (MSE) of the estimated vector. For a fixed $\mathbf{x}\in\mathcal{X}_K$ and $k\in\{1,\ldots,K\}$, define the set of difference vectors
\begin{equation}
\mathcal{D}_k(\mathbf{x}) \triangleq
\left\{
\begin{aligned}
&\sum_{i\in S_+}\mathbf{e}_i-\sum_{j\in S_-}\mathbf{e}_j
\ \bigg|\
|S_+|=|S_-|=k,\\
&\mathbf{x}+\boldsymbol{\delta}\in\mathcal{X}_K
\end{aligned}
\right\}.
\label{eq:delta_family}
\end{equation}
Each vector in $\mathcal{D}_k(\mathbf{x})$ moves $k$ occupied voxels to $k$ empty positions, so $\mathbf{x}$ and $\mathbf{x}+\boldsymbol{\delta}$ differ in exactly $k$ occupied voxels. In the ZZB literature, this set is called the $k$-th shell. {Here, $S_-\subseteq\{n:x_n=1\}$ indexes the $k$ occupied voxels removed from the support, whereas $S_+\subseteq\{n:x_n=0\}$ indexes the $k$ empty voxels added to it.} The number of such vectors is $C_k \triangleq |\mathcal{D}_k(\mathbf{x})|=\binom{K}{k}\binom{N-K}{k},$ which is independent of $\mathbf{x}$. The corresponding ZZB depends on the average minimum error probability over this shell, as
\begin{equation}
\bar P_k(\mathbf{A},\sigma^2)
\triangleq
\frac{1}{|\mathcal{X}_K|C_k}
\sum_{\mathbf{x}\in\mathcal{X}_K}
\sum_{\boldsymbol{\delta}\in\mathcal{D}_k(\mathbf{x})}
P_{\min}(\mathbf{x},\boldsymbol{\delta}),
\label{eq:shell_average}
\end{equation}
where $P_{\min}(\mathbf{x},\boldsymbol{\delta})$ is the minimum error probability of deciding between $\mathbf{x}$ and $\mathbf{x}+\boldsymbol{\delta}$~\cite{jeong2025zzb}.

\begin{proposition}
\label{prop:shell_zzb}
\emph{When $\mathbf{x}$ is uniformly distributed over $\mathcal{X}_K$, every estimator $\hat{\mathbf{x}}\in\mathcal{X}_K$ satisfies
\begin{equation}
\mathbb{E}\!\left[\frac{J(\hat{\mathbf{x}},\mathbf{x})}{K}\right]
\ge
\mathcal{Z}
\triangleq
\frac{1}{K}\sum_{k=1}^{\lceil K/2\rceil}
\mathcal{V}\!\left\{\widetilde P_k(\mathbf{A},\sigma^2)\right\},
\label{eq:global_shell_zzb}
\end{equation}
where $\mathcal{Z}$ denotes the ZZB, $\mathcal{V}\{\cdot\}$ is the valley-filling operator, and for $k=1,\ldots,\left\lceil\frac{K}{2}\right\rceil,$
\begin{equation}
\widetilde P_k(\mathbf{A},\sigma^2)
\triangleq
\max\!\left\{
\bar P_{2k-1}(\mathbf{A},\sigma^2),
\bar P_{2k}(\mathbf{A},\sigma^2)
\right\},
\label{eq:grouped_shell_average}
\end{equation}
with the convention $\bar P_{K+1}(\mathbf{A},\sigma^2)\equiv 0$ when $K$ is odd. Consequently,
\begin{equation}
\mathbb{E}\!\left[\frac{\|\hat{\mathbf{x}}-\mathbf{x}\|_2^2}{K}\right]
\ge2
\mathcal{Z}.
\label{eq:global_shell_nmse}
\end{equation}}
\end{proposition}
\begin{proof}
    The proof is provided in Appendix~\ref{ap:zzb}.
\end{proof}
Proposition~\ref{prop:shell_zzb} establishes the ZZB for the normalized reconstruction error in terms of the shell-averaged pairwise error probabilities $\bar P_k(\mathbf{A},\sigma^2)$. Therefore, evaluating the ZZB reduces to computing $\bar P_k(\mathbf{A},\sigma^2)$, which is determined by the pairwise distance induced by the measurement matrix.

To derive this quantity, we consider an admissible perturbation $\boldsymbol{\delta}\in\mathcal{D}_k(\mathbf{x})$. Then, the equal-prior binary test is $\mathcal H_0:\tilde{\mathbf y}\sim\mathcal{CN}(\mathbf A\mathbf x,\sigma^2\mathbf I)$ versus $\mathcal H_1:\tilde{\mathbf y}\sim\mathcal{CN}(\mathbf A(\mathbf x+\boldsymbol\delta),\sigma^2\mathbf I)$. Since the two hypotheses have the same covariance, the exact maximum a posteriori (MAP) error is $P_{\min}(\mathbf{x},\boldsymbol{\delta})
\!\!=\!\!
Q\!\left(\frac{\|\mathbf{A}\boldsymbol{\delta}\|_2}{\sqrt{2}\sigma}\right)$,
where the Gaussian $Q$-function is $Q(a)\!\triangleq\!\! \frac{1}{\sqrt{2\pi}}\int_{a}^{\infty}\!e^{-t^2/2}dt\!
=\!\!\frac{1}{2}\operatorname{erfc}\!\left(\!\frac{a}{\sqrt{2}}\!\right).$ With $\mathbf{G}=\mathbf{A}^H\mathbf{A}$,
\begin{equation}
P_{\min}(\mathbf{x},\boldsymbol{\delta})
=
\frac{1}{2}\operatorname{erfc}\!\left(\sqrt{\frac{\boldsymbol{\delta}^T\mathbf{G}\boldsymbol{\delta}}{4\sigma^2}}\right).
\label{eq:pairwise_pe}
\end{equation}
Substituting \eqref{eq:pairwise_pe} into \eqref{eq:shell_average} yields
\begin{equation}
\begin{aligned}
\bar P_k(\mathbf{A},\sigma^2)
=
\frac{1}{|\mathcal{X}_K|C_k}
\sum_{\mathbf{x}\in\mathcal{X}_K}
\sum_{\boldsymbol{\delta}\in\mathcal{D}_k(\mathbf{x})}P_{\min}.
\end{aligned}
\label{eq:shell_average_closed}
\end{equation}
The remaining task is therefore to characterize $\|\mathbf{A}\boldsymbol{\delta}\|_2^2$.

Under a common propagation model, both architectures share the same steering block $\mathbf{S}_{t,u}$ and differ only through the two-chain mixing matrix $\mathbf{W}_t$. Set $
s_{t,u}^{F}(\boldsymbol{\delta})\triangleq [\mathbf{S}_{t,u}\boldsymbol{\delta}]_1,
s_{t,u}^{B}(\boldsymbol{\delta})\triangleq [\mathbf{S}_{t,u}\boldsymbol{\delta}]_2$
as the perturbation responses along the forward and backward PA paths. Then
\begin{equation}
\|\mathbf{A}\boldsymbol{\delta}\|_2^2
=
\sum_{t=1}^{T}\sum_{u=1}^{U}D_{t,u}(\boldsymbol{\delta}),
\label{eq:block_distance}
\end{equation}
where
\begin{equation}
D_{t,u}(\boldsymbol{\delta})
=
\sum_{m=1}^{2}\left|w_{m1}^{(t)}s_{t,u}^{F}(\boldsymbol{\delta})+w_{m2}^{(t)}s_{t,u}^{B}(\boldsymbol{\delta})\right|^{2}.
\label{eq:generic_distance}
\end{equation}
where $w_{mk}^{(t)}$ is the $(m,k)$-th entry of $\mathbf{W}_t$.

For end-fed PASS, $y_{e,1}=y_{e,2}=0$, substituting \eqref{eq:mixing} into \eqref{eq:generic_distance} yields
\begin{equation}
\psi_{t,u}^{\mathrm{End}}(\boldsymbol{\delta})
\triangleq
e^{-j\beta_g|y_t^F|}s_{t,u}^{F}(\boldsymbol{\delta})
+
e^{-j\beta_g|y_t^B|}s_{t,u}^{B}(\boldsymbol{\delta}),
\end{equation}
\begin{equation}
D_{t,u}^{\mathrm{End}}(\boldsymbol{\delta})
=
2\left|\psi_{t,u}^{\mathrm{End}}(\boldsymbol{\delta})\right|^2.
\label{eq:endfed_distance}
\end{equation}
Each user-slot block therefore contributes only one coherent mixture of the forward and backward perturbation responses.

For C-PASS, based on \eqref{eq:mixing}, for $r\in\{1,2\}$,
\begin{equation}
\psi_{t,u,r}^{\mathrm{Cen}}(\boldsymbol{\delta})
\!\triangleq\!
e^{-j\beta_g|y_t^F\!-y_{c,r}|}s_{t,u}^{F}(\boldsymbol{\delta})
\!+\!
e^{-j\beta_g|y_t^B-y_{c,r}|}s_{t,u}^{B}(\boldsymbol{\delta}),
\end{equation}
\begin{equation}
D_{t,u}^{\mathrm{Cen}}(\boldsymbol{\delta})
=
\left|\psi_{t,u,1}^{\mathrm{Cen}}(\boldsymbol{\delta})\right|^2
+
\left|\psi_{t,u,2}^{\mathrm{Cen}}(\boldsymbol{\delta})\right|^2.
\label{eq:cpass_distance_exact}
\end{equation}
Using \eqref{eq:det_cpass}, the eigenvalues of
$(\mathbf{W}_t^{\mathrm{Cen}})^H\mathbf{W}_t^{\mathrm{Cen}}$ are
\begin{equation}
\lambda_{\max}\!\left((\mathbf{W}_t^{\mathrm{Cen}})^H\mathbf{W}_t^{\mathrm{Cen}}\right)=4\cos^2\!\left(\frac{\beta_g d}{2}\right),
\end{equation}
\begin{equation}
\lambda_{\min}\!\left((\mathbf{W}_t^{\mathrm{Cen}})^H\mathbf{W}_t^{\mathrm{Cen}}\right)=4\sin^2\!\left(\frac{\beta_g d}{2}\right).
\label{eq:cpass_lambda}
\end{equation}
Hence
\begin{equation}
D_{t,u}^{\mathrm{Cen}}(\boldsymbol{\delta})
\ge
4\sin^2\!\left(\frac{\beta_g d}{2}\right)
\Big(
|s_{t,u}^{F}(\boldsymbol{\delta})|^2+
|s_{t,u}^{B}(\boldsymbol{\delta})|^2
\Big).
\label{eq:cpass_distance_lb}
\end{equation}

\begin{figure*}[t]
  \centering
  \begin{subfigure}[t]{0.3\textwidth}
    \centering
    \includegraphics[width=\linewidth]{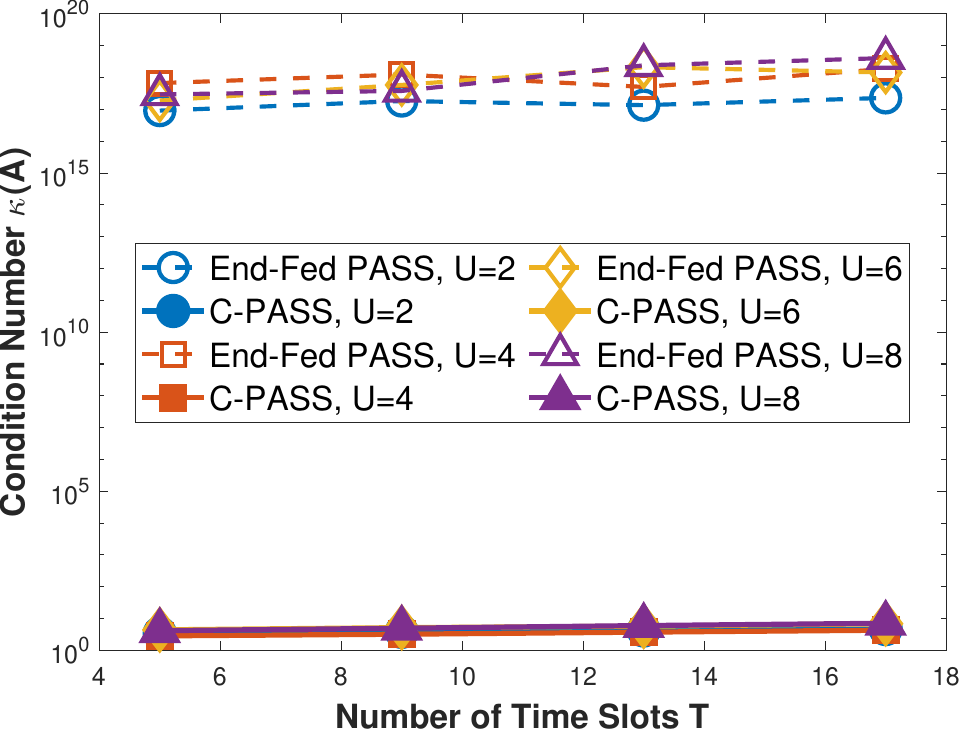}
\caption{Condition number versus number of time slots.}\label{fig:condition_number}
  \end{subfigure}
  \begin{subfigure}[t]{0.3\textwidth}
    \centering
    \includegraphics[width=\linewidth]{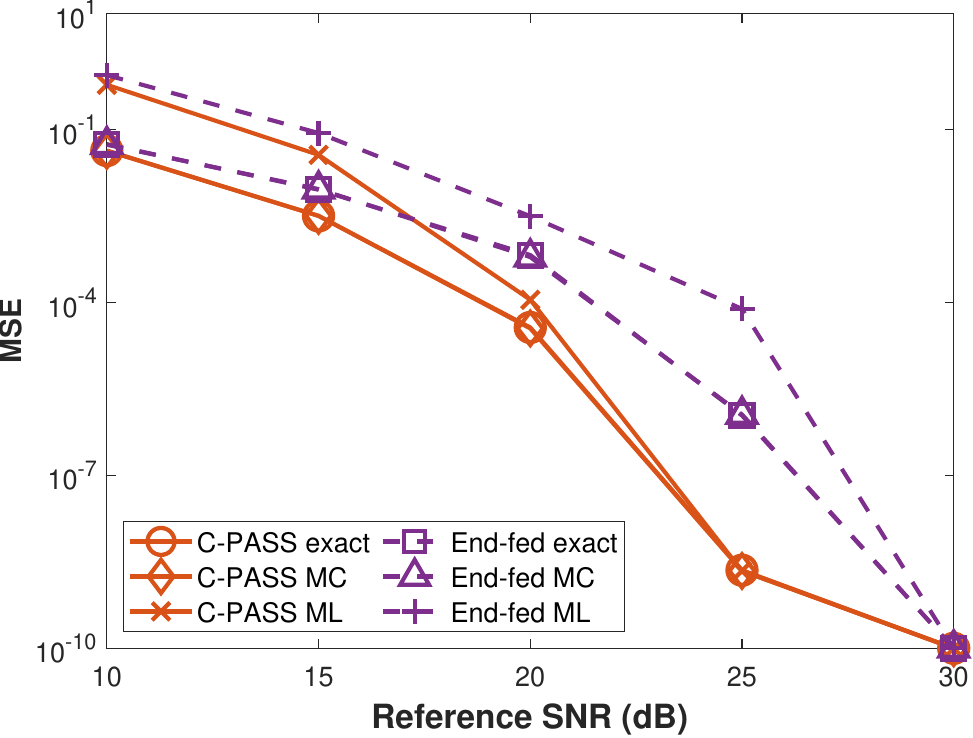}
    \caption{ZZB validation via SNR.}\label{fig:zzb_snr_phase}
  \end{subfigure}
  \begin{subfigure}[t]{0.3\textwidth}
    \centering
    \includegraphics[width=\linewidth]{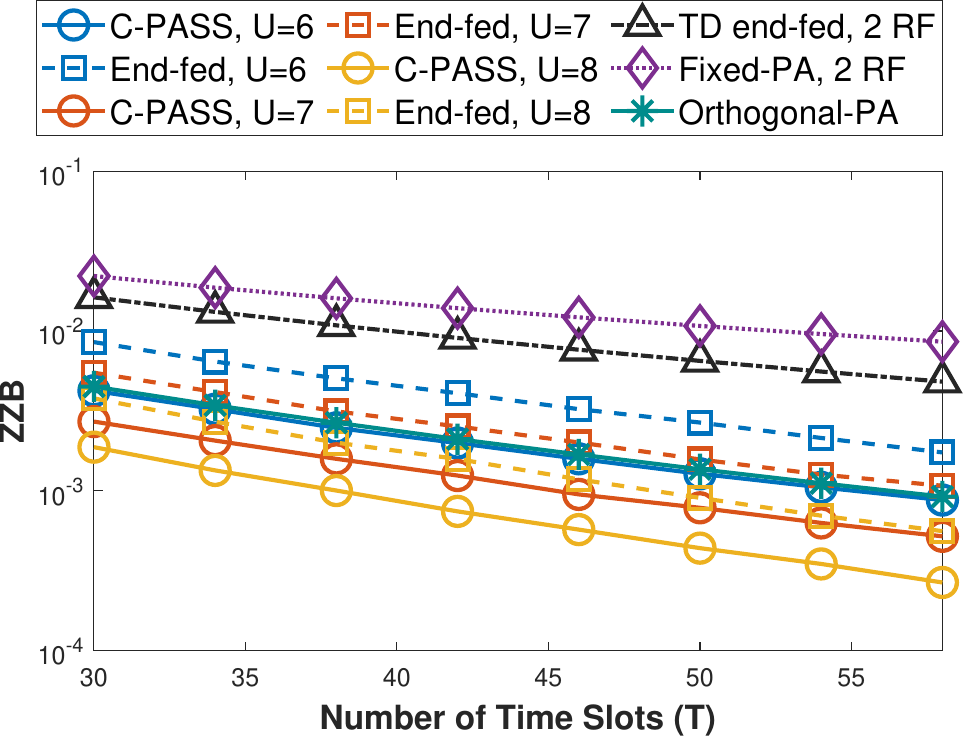}
    \caption{ZZB vs time slots.}\label{fig:zzb_timeslots}
  \end{subfigure}
  \caption{Numerical results of C-PASS vs End-fed PASS.}
  \label{fig:three-subfig}
\end{figure*}

Substituting \eqref{eq:endfed_distance} and \eqref{eq:cpass_distance_exact} into \eqref{eq:shell_average_closed} yields the final architecture-specific closed form. For $\chi\in\{\mathrm{End},\mathrm{Cen}\}$,
\begin{equation}
\begin{aligned}
\bar P_k(\mathbf{A}^{\chi},\sigma^2)
=&
\frac{1}{2|\mathcal{X}_K|C_k}
\sum_{\mathbf{x}\in\mathcal{X}_K}
\sum_{\boldsymbol{\delta}\in\mathcal{D}_k(\mathbf{x})}\\
&\times
\operatorname{erfc}\!\left(
\sqrt{
\frac{
\sum_{t=1}^{T}\sum_{u=1}^{U}
D_{t,u}^{\chi}(\boldsymbol{\delta})
}{4\sigma^2}}
\right).
\end{aligned}
\label{eq:shell_average_arch}
\end{equation}

\begin{remark}
\emph{Equations~\eqref{eq:cpass_distance_exact}, \eqref{eq:cpass_distance_lb}, and \eqref{eq:shell_average_arch} reveal the key ZZB difference between the two architectures. For end-fed PASS, $\lambda_{\min}((\mathbf{W}_t^{\mathrm{End}})^H\mathbf{W}_t^{\mathrm{End}})=0$, so each pair is evaluated through a single coherent mixture of the two received signals. In contrast, C-PASS uses two feed-point measurements with different mixing coefficients, which yields a lower ZZB under the same sensing budget.}
\end{remark}

{Constructing and storing $\mathbf A\in\mathbb C^{2UT\times N}$ costs $\mathcal O(UTN)$. Exact shell-$k$ evaluation uses $\binom NK\binom Kk\binom{N-K}{k}$ pairs at $\mathcal O(UTk)$ per pair, and exhaustive ML costs $\mathcal O(UTK\binom NK)$; both are combinatorial. Sampling $L$ perturbations per shell costs $\mathcal O\!\left(UTL\sum_{k=1}^{K}k\right)=\mathcal O(UTLK^2)=\mathcal O(UTL\rho^2N^2)$ for $\rho=K/N$, with $N$ governing matrix construction and storage and $K$ governing shell count and per-sample work. First-order recovery requires $\mathcal O(UTN)$ matrix-vector work per iteration.}

% \begin{revisionblock}
% For $M$ feeds and $P$ active PA paths, $\mathbf A_{t,u}=\mathbf W_t\mathbf S_{t,u}$ with $\mathbf W_t\in\mathbb C^{M\times P}$, so the Gram-matrix and pairwise-distance ZZB analyses extend directly, and independent multi-waveguide blocks are vertically stacked with additive Gram contributions. The closed-form spacing in Proposition~\ref{prop:conditioning} is specific to symmetric two-feed/two-path mixing; general layouts instead require optimizing the singular values of $\mathbf W_t$. Regarding hardware, the two simulated end-fed outputs are co-located and identical, so the second RF chain adds no sensing DoF and a minimal receiver could omit it. C-PASS obtains independently phased observations within each slot from two spatially separated feeds and two calibrated RF chains, whereas time-division end-fed reception avoids the separated-feed hardware but needs two pilot subslots per PA pair, trading simultaneous spatial information for serial sounding time.\end{revisionblock}

\section{Numerical Results}
This section provides numerical results to corroborate the analysis and compare C-PASS with conventional end-fed PASS in terms of measurement conditioning and the resulting ZZB under different system settings.

\subsection{Simulation Setup}

Unless otherwise specified, the simulations use a carrier frequency of $f_c=28$~GHz and a waveguide effective refractive index of $n_e=1.4$. {Since the practical minimum spacing is about $\lambda/2$ \cite{number}, $S=0$ is infeasible and $S=1$ offers little clearance; $S=4$ is a conservative optimally conditioned choice, giving $d=9\lambda_g/4$.} The region of interest is a $5\times10\times3$~m$^3$ cuboid, discretized into $0.5\times0.5\times0.5$~m$^3$ voxels, resulting in $N=1200$ voxels with a sparsity ratio of $0.5\%$. A single waveguide is deployed at height $D_h=3$~m along the $y$-axis. The number of UEs is $U=6$, and each UE transmits with power $P_u=1$~W using orthogonal pilots, {with noise variance $\sigma^2=10^{-12}$~W.} The UE positions are randomly generated within the coverage area at ground level.

\subsection{Measurement Matrix Conditioning}
Fig.~\ref{fig:condition_number} plots $\kappa(\mathbf{A})$ versus $T$ for both architectures under $U\in\{2,4,6,8\}$. {For end-fed PASS, duplicated co-located-feed rows make $\mathbf A^{\mathrm{End}}$ rank deficient with $\sigma_{\min}=0$ and theoretical condition number $+\infty$; values near $10^{16}$ are finite-precision saturation and imply that reliable inversion is numerically impossible.} In contrast, C-PASS maintains $\kappa(\mathbf{A}^{\mathrm{Cen}})\lesssim 10$ across all tested configurations. {Thus, Fig.~\ref{fig:condition_number} confirms at the full-$\mathbf A$ level the conditioning behavior predicted by $\mathbf W_t$.}

\subsection{Validation and Comparison of the Derived ZZB}
Fig.~\ref{fig:zzb_snr_phase} validates Proposition~\ref{prop:shell_zzb} on an exactly enumerable $12$-voxel, $5$-active Bayesian subproblem, {which simplifies the local exhaustive calculation.} Here Exact denotes the direct evaluation of \eqref{eq:global_shell_zzb}, MC its Monte Carlo approximation, and ML half the empirical MSE of the maximum-likelihood estimator. The Exact and MC curves are nearly indistinguishable, confirming the numerical implementation of \eqref{eq:global_shell_zzb}, while the ML curve remains above the corresponding ZZB for both architectures, consistent with the ZZB being a valid Bayesian lower bound. The C-PASS ZZB and ML curves stay strictly below those of end-fed PASS throughout. {For $N=1200$, the same expression is evaluated using sampled shell perturbations.}

\subsection{ZZB Performance versus Number of Time Slots}
Fig.~\ref{fig:zzb_timeslots} plots the {nine-curve} comparison of the ZZB versus $T$. For all $U\in\{6,7,8\}$, C-PASS achieves a lower ZZB than end-fed PASS over the entire range, in agreement with \eqref{eq:block_distance}-\eqref{eq:cpass_distance_lb}, and the gap remains nearly unchanged as $U$ increases, indicating that the conditioning advantage is preserved with more users. {For $U=6$, Fixed-PA repeats one PA, TD end-fed senses the two PA components serially, and $\mathbf W=\sqrt{2}\mathbf I$ is an idealized reference. Their higher ZZBs show that the C-PASS gain comes from separated feeds rather than RF-chain count or serial samples, while the small reference gap indicates near-ideal component decoupling on a shared waveguide.}

\section{Conclusion}
A C-PASS-enabled uplink environment sensing framework has been proposed in this letter. For the considered binary voxel-map reconstruction problem, we characterize the feed-point distance for stable separation of the received signals and derive ZZB expressions for both C-PASS and conventional end-fed PASS. The analysis shows that C-PASS improves the conditioning of the sensing measurement matrix and yields a lower reconstruction-error bound under the same sensing budget. {Future work should investigate the impact of relaxing the current assumptions, including the single-waveguide topology, known $K$, narrowband propagation, independent AWGN, perfect LoS removal, known UE locations, and orthogonal pilots.}
\appendices
\section{Proof of Proposition~\ref{prop:shell_zzb}}
\label{ap:zzb}
Fix $k\in\{1,\ldots,K\}$ and an equiprobable pair $(\mathbf{x},\mathbf{x}')\in\mathcal{X}_K^2$ with $J(\mathbf{x}',\mathbf{x})=k$. Let the estimator induce the nearest-neighbor binary test and declare $H_0$ if $J(\hat{\mathbf{x}},\mathbf{x})\le J(\hat{\mathbf{x}},\mathbf{x}')$, and $H_1$ otherwise. If $H_0$ is true and $J(\hat{\mathbf{x}},\mathbf{x})<k/2$, the triangle inequality gives $J(\hat{\mathbf{x}},\mathbf{x}')\ge k-J(\hat{\mathbf{x}},\mathbf{x})>J(\hat{\mathbf{x}},\mathbf{x})$, so the test is correct. Since $J$ is integer-valued, an error therefore implies $J(\hat{\mathbf{x}},\mathbf{x})\ge\lceil k/2\rceil$ under $H_0$, and similarly under $H_1$. Thus,
\begin{equation}
\begin{aligned}
P_{\min}(\mathbf{x},\mathbf{x}'-\mathbf{x})
\le&
\frac{1}{2}\Pr\!\left(
J(\hat{\mathbf{x}},\mathbf{x})\ge \left\lceil\frac{k}{2}\right\rceil
\,\middle|\, \mathbf{x}
\right)
\\
&+
\frac{1}{2}\Pr\!\left(
J(\hat{\mathbf{x}},\mathbf{x}')\ge \left\lceil\frac{k}{2}\right\rceil
\,\middle|\, \mathbf{x}'
\right).
\end{aligned}
\label{eq:pmin_pair_tail}
\end{equation}
Write $\mathbf{x}'=\mathbf{x}+\boldsymbol{\delta}$ with $\boldsymbol{\delta}\in\mathcal{D}_k(\mathbf{x})$ and average \eqref{eq:pmin_pair_tail} over $(\mathbf{x},\boldsymbol{\delta})\in\mathcal{X}_K\times\mathcal{D}_k(\mathbf{x})$. By \eqref{eq:shell_average}, this yields
\begin{equation}
\bar P_k(\mathbf{A},\sigma^2)
\le
\frac{S_{1,k}+S_{2,k}}{2|\mathcal{X}_K|C_k},
\label{eq:avg_split_short}
\end{equation}
where, for $\mathbf{z}_1=\mathbf{x}$ and $\mathbf{z}_2=\mathbf{x}+\boldsymbol{\delta}$,
\begin{equation}
S_{i,k}
\triangleq
\sum_{\mathbf{x}\in\mathcal{X}_K}\!
\sum_{\boldsymbol{\delta}\in\mathcal{D}_k(\mathbf{x})}\!\!\!\!
\Pr\!\left(J(\hat{\mathbf{x}},\mathbf{z}_i)\ge \left\lceil\frac{k}{2}\right\rceil \,\middle|\, \mathbf{z}_i\right),
i=1,2.
\label{eq:sik_def}
\end{equation}
Here $S_{1,k}$ is independent of $\boldsymbol{\delta}$ and contributes $|\mathcal{X}_K|C_k$ identical terms, while $S_{2,k}$ reduces to the same quantity under the bijection $(\mathbf{x},\boldsymbol{\delta})\leftrightarrow(\mathbf{x}+\boldsymbol{\delta},-\boldsymbol{\delta})$. Hence, by \eqref{eq:grouped_shell_average} and \eqref{eq:avg_split_short}-\eqref{eq:sik_def}, for $k=1,\ldots,\lceil K/2\rceil$,
\begin{equation}
\bar P_k(\mathbf{A},\sigma^2)
\le
\Pr\!\left(J(\hat{\mathbf{x}},\mathbf{x})\ge \left\lceil\frac{k}{2}\right\rceil\right).
\label{eq:tail_shell_lower}
\end{equation}
Since $\Pr(J(\hat{\mathbf{x}},\mathbf{x})\ge k)$ is nonincreasing in $k$, the valley-filling operator preserves \eqref{eq:tail_shell_lower}, i.e., $\mathcal{V}\{\widetilde P_k(\mathbf{A},\sigma^2)\}\le\Pr(J(\hat{\mathbf{x}},\mathbf{x})\ge k)$, and summing over the shells gives
\begin{equation}
\mathbb{E}\!\left[J(\hat{\mathbf{x}},\mathbf{x})\right]
\ge
\sum_{\ell=1}^{\lceil K/2\rceil}
\mathcal{V}\!\left\{\widetilde P_\ell(\mathbf{A},\sigma^2)\right\}.
\label{eq:grouped_sum_lb}
\end{equation}
Dividing \eqref{eq:grouped_sum_lb} by $K$ gives \eqref{eq:global_shell_zzb}, and applying \eqref{eq:nmse_relation} then yields $\mathbb{E}[\|\hat{\mathbf{x}}-\mathbf{x}\|_2^2/K]=2\,\mathbb{E}[J(\hat{\mathbf{x}},\mathbf{x})/K]\ge 2\,\mathcal{Z}$, which is exactly \eqref{eq:global_shell_nmse}.
\vspace{-0.5cm}

\bibliographystyle{IEEEtran}

\bibliography{cas_refs}

\end{document}